\documentclass[journal,a4paper]{IEEEtran}

\usepackage{url}
\usepackage{amsfonts}
\usepackage{amsmath}
\usepackage{amssymb}
\usepackage{ifpdf}
\usepackage{cite}
\usepackage{color}
\usepackage{comment}

\usepackage{enumitem}

\usepackage[textsize=scriptsize,textwidth=\marginparwidth]{todonotes}

\usepackage{amsthm}

\newtheorem{theorem}{Theorem}[section]

\newtheorem{corollary}[theorem]{Corollary}

\theoremstyle{definition}

\usepackage[normalem]{ulem}

\newcommand{\calM}{\mathcal{M}}
\newcommand{\calS}{\mathcal{S}}

\newcommand{\vm}{\vec{m}}

\newcommand{\vv}{\vec{v}}

\newcommand{\Z}{\mathbb{Z}}
\newcommand{\N}{\mathbb{N}}
\newcommand{\R}{\mathbb{R}}

\newcommand{\setr}[2]{\left\{\ #1 \ \left|\ #2 \right. \ \right\}}

\usepackage{centernot}

\usepackage{scalerel}

\title{Design of Geometric Molecular Bonds}

\author{David Doty and Andrew Winslow
\thanks{A preliminary draft of this article appeared as~\cite{dgmbISIT}.}
\thanks{D. Doty was supported by NSF grants 1219274, 1619343, and the Molecular Programming Project under NSF grant 1317694. Both authors were supported by NSF grant 1422152.}
\thanks{D. Doty is with the Computer Science department of the University of California, Davis, {\tt doty@ucdavis.edu}.}
\thanks{A. Winslow is with the Computer Science department of the University of Texas Rio Grande Valley, {\tt andrew.winslow@utrgv.edu}.}}

\date{}

\begin{document}

\maketitle

\begin{abstract}
An example of a \emph{specific} molecular bond is the affinity of the DNA base A for T, but not for C, G, or another A.
This contrasts \emph{nonspecific} bonds,
such as the affinity of any positive charge for any negative charge (like-unlike),
or of nonpolar material for itself when in aqueous solution (like-like).

Recent experimental breakthroughs in DNA nanotechnology~\cite{sungwook_bonds_2011,dietz_bonds_science_2015} demonstrate that a particular nonspecific like-like bond (``blunt-end DNA stacking'' that occurs between the ends of any pair of DNA double-helices) can be used to create specific ``macrobonds'' by careful geometric arrangement of many nonspecific blunt ends, motivating the need for sets of macrobonds that are \emph{orthogonal}: two macrobonds not intended to bind have relatively low binding strength, even when misaligned.

To address this need, we introduce \emph{geometric  orthogonal codes} that abstractly model the engineered DNA macrobonds as two-dimensional binary codewords.
While motivated by completely different applications, geometric orthogonal codes share similar features to the \emph{optical orthogonal codes} studied by Chung, Salehi, and Wei~\cite{optical_orthogonal_codes_1989}.
The main technical difference is the importance of 2D geometry in defining codeword orthogonality.
\end{abstract}



\pagenumbering{gobble}

\section{Introduction}\label{sec:intro}

\subsection{Structural DNA nanotechnology}
DNA nanotechnology began in the 1980s when Seeman~\cite{Seem82} showed that artificially synthesized DNA strands could be designed to automatically self-assemble nanoscale structures, rationally designed through the choice of DNA sequences.
In the past 20 years, the field has witnessed a dramatic surge in the development of basic science, \emph{in vitro} applications, such as chemical oscillators and molecular walkers, and \emph{in vivo} applications, such as drug delivery, cellular RNA sensing, and genetically encoded structures~\cite{dnaNanoSurveySeelig2015}.

A technological pillar of the field is \emph{DNA origami}, developed by Rothemund~\cite{rothemund2006folding}, a simple, fast, inexpensive, and reliable method for creating artificial 2D and 3D DNA structures, with a control resolution of a few nanometers.
DNA origami requires a single long \emph{scaffold} strand of DNA; the most commonly used is the 7249-nucleotide single-stranded genome of the bacteriophage virus M13mp18, widely and cheaply available from many biotech companies.
The scaffold is mixed with a few hundred shorter ($\approx 32$nt) synthesized DNA strands called \emph{staples}, each of which is designed to bind (through Watson-Crick complementarity) to 2-3 regions of the scaffold.
Via thermal annealing, the staples fold the scaffold strand into a shape dictated by the choice of staple DNA sequences, hence the term \emph{origami}.
The process is illustrated in Figure~\ref{fig:origami}(a), with the results shown in Figure~\ref{fig:origami}(b).

\begin{figure*}[ht]
  \begin{centering}
    \includegraphics[width=7.2in]{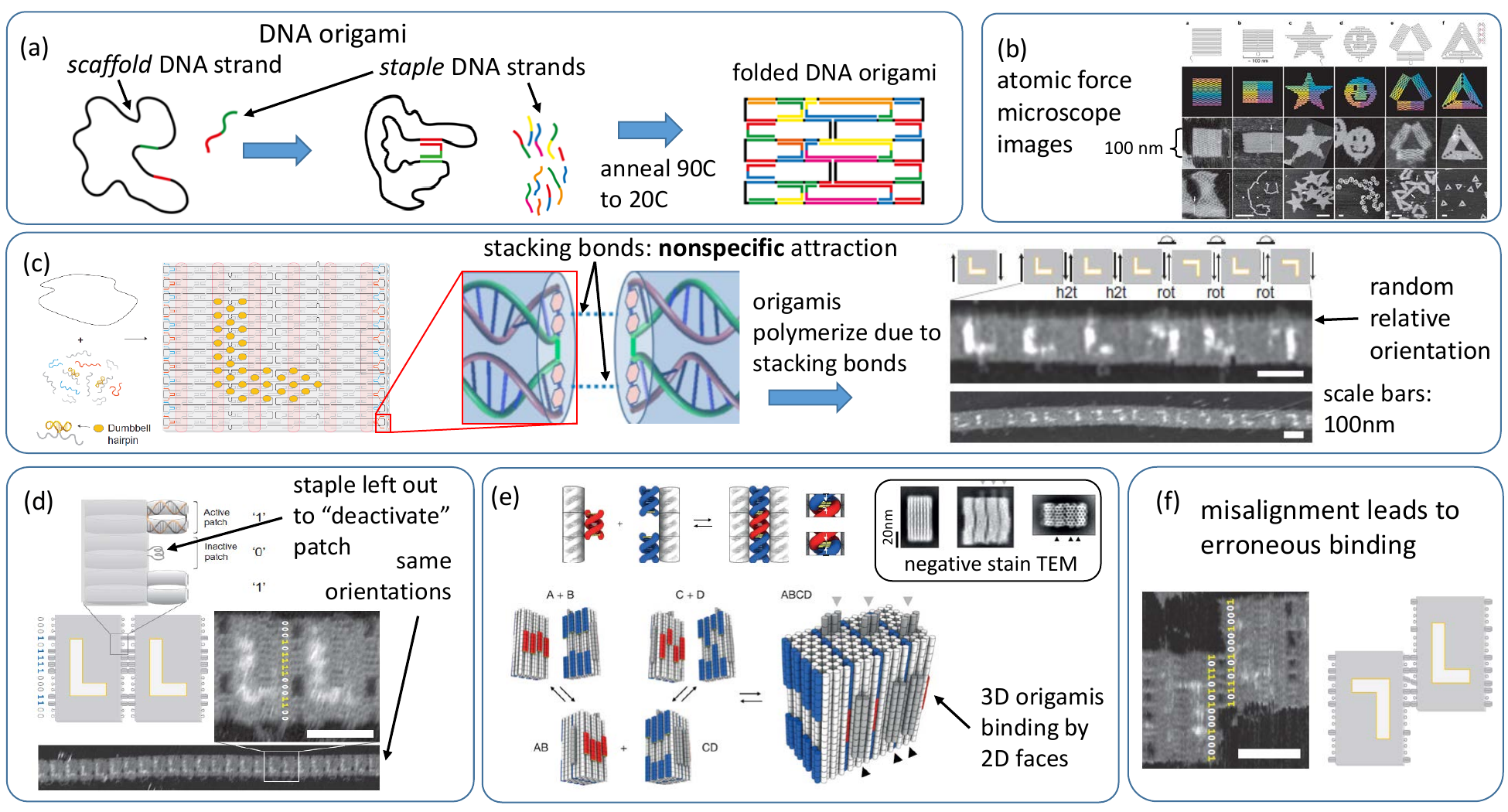}
    \caption{
      Illustration of DNA origami and geometrically programmable stacking bonds.
      {\bf (a)} DNA origami illustration (source: \texttt{http://openwetware.org/wiki/Biomod/2014/Design}).
      {\bf (b)} Atomic force microscope images of nanoscale shapes assembled by DNA origami technique (source:~\cite{rothemund2006folding}).
      {\bf (c)} Stacking bonds are nonspecific attraction occurring between the \emph{ends} of two DNA helices, such as those that appear on the edges of a DNA origami causing origamis to form long chains (\emph{polymerize}) in solution.
      Markers on the origami surface (an asymmetrical $L$ shape) reveal that the orientation of origamis in a polymer is random; i.e., some are ``upside down'' relative to others (source: \texttt{http://openwetware.org/wiki/Biomod/2014/Kansai/Experiment} and~\cite{sungwook_bonds_2011}).
      {\bf (d)} By removing certain stacking bonds at specific locations to create a binary pattern on the edge of a DNA origami, the whole edge becomes a specific ``macrobond'' that binds most strongly to another edge with the same pattern; in this case, the left side of an origami binds favorably to the right side, and less favorably to another left side (source:~\cite{sungwook_bonds_2011}).
      {\bf (e)} The technique also works to bind 3D origami using 2D patterns of stacking bonds on their faces, using stacking bonds in a slightly different way than in part (d).
      The placement of the nonspecific bonds gives the entire face a higher affinity for another face with a complementary pattern.
      Gray and black triangular arrows indicate bumps placed on origamis to allow verification in TEM images that the four monomers are in fact A/B/C/D bound as intended. (source:~\cite{dietz_bonds_science_2015}, Reprinted with permission from AAAS).
      {\bf (f)} One source of error is the matching of many stacking bonds between two misaligned faces.
      Note that this image shows a second source of error: the ``bending'' of one DNA helix to bind to another (such errors are not modeled in this paper). (source:~\cite{sungwook_bonds_2011}).
      openwetware images licensed under Creative Commons Attribution-ShareAlike 3.0 Unported: \texttt{http://www.openwetware.org/wiki/OpenWetWare:Copyright}
      Images from \cite{sungwook_bonds_2011}, \texttt{http://www.nature.com/nchem/journal/v3/n8/abs/nchem.1070.html},
      and~\cite{rothemund2006folding}, \texttt{http://www.nature.com/nature/journal/v440/n7082/abs/nature04586.html},
      Reprinted with permission from NPG)
    }
    \label{fig:origami}
  \end{centering}
  \vspace{-0.3cm}
\end{figure*}

Although the Watson-Crick pairing of bases between two single strands of DNA is very specific, DNA is known to undergo other, less specific interactions.
One well-studied interaction is called a \emph{stacking bond}, formed when \emph{any} pair of terminated double helices --- known as \emph{blunt ends} --- face each other, as shown in Figure~\ref{fig:origami}(c).
Since two edges of a standard DNA origami rectangle consist entirely of blunt ends, DNA origami rectangles are known to bind along their edges to form long polymers of many origamis, despite the fact that no hybridization between single strands occurs between them.
One way to avoid stacking between origamis is to leave out staple strands along the edge, so that rather than blunt ends, there are single-stranded loops of the scaffold strand~\cite{rothemund2006folding}.

Woo and Rothemund~\cite{sungwook_bonds_2011} turned the bug of unintended origami stacking into a feature with the following idea: leave out \emph{some} of the staples along the edge, but keep others; see Figure~\ref{fig:origami}(d).
Although individual blunt ends bind nonspecifically to others, the only way for \emph{all} blunt ends along an edge to bind is with matching blunt ends on another origami in the same relative positions.
Thus, geometric placement of blunt ends makes the entire side of an origami into a specific ``macrobond''.
Figure~\ref{fig:origami}(d) shows how this approach enforces that a set of origamis bind to form only intended arrangements.

The idea extends from 2D origami rectangles with 1D edges, to 3D origami boxes with 2D rectangular faces as demonstrated by Gerling, Wagenbauer, Neuner, and Dietz~\cite{dietz_bonds_science_2015}.\footnote{The way stacking bonds are used is slightly different. Rather than helices orthogonal to the origami face, they are parallel. As seen in Figure~\ref{fig:origami}(e), pairs of whole helices protrude and a complementary face has a two-helix ``gap'' into which these helices fit to form four total stacking bonds.}
They rationally design polymers of many origamis with prescribed sizes and shapes such as the 4-mer ABCD shown in Fig.~\ref{fig:origami}(e).

The preceding description of macrobonds is idealized: other mechanisms may permit unintended pairs of macrobonds to bind spuriously.
Figure~\ref{fig:origami}(f) shows two macrobonds aligning sufficiently many blunt ends to attach stably via two such mechanisms: \emph{flexibility} of DNA helices and \emph{misalignment} of macrobonds.
This paper is an attempt to attack the latter problem with coding theory.

\subsection{Definitions and main result}
\label{sec:intro:statement-of-result}

\newcommand{\flip}{\mathsf{rot180}}

Although inspired by work in DNA nanotechnology, the design of specific macrobonds formed by geometric arrangements of nonspecific bonds is fundamental and likely to be part of the future of nanotechnology, even if based on substrates other than DNA.
We abstract away several details of DNA origami in mathematically formulating the problem.
This subsection,~\ref{sec:intro:statement-of-result}, simply states the formal definitions and main result,
while subsection~\ref{sec:intro:relationship-defn-experiments} discusses the relationship between the definitions and the experimental motivation.

Let $[n]=\{0,1,\ldots,n-1\}$.
We model each 2D face of a monomer (e.g., a DNA origami) as a discrete $n \times n$ square $[n]^2$, with $n$ representing the placement resolution of nonspecific bonds, called \emph{patches}.
A \emph{macrobond} is a subset
$M \subseteq [n]^2$.
Given a vector $\vv \in \Z^2$, $M+\vv = \setr{\vm+\vv}{\vm \in M}$ denotes $M$ translated by $\vv$.\footnote{The translation operation $M + \vv$ is normal translation in $\Z^2$; there is no ``wrapping'' of points that are shifted beyond the edge of a macrobond around to the other edge, as in optical orthogonal codes (discussed in Section~\ref{sec:related}).
Thus, although $M \subseteq [n]^2$, potentially $(M+\vec{v}) \not\subseteq [n]^2$.}
A parameter $w \in \{2,\ldots,n^2\}$ denotes the \emph{macrobond strength} (a.k.a., \emph{codeword weight}).
A parameter $\lambda \in \{1,\ldots,w-1\}$ denotes the \emph{mismatch strength limit}.
An \emph{$(n,w,\lambda)$-geometric  orthogonal code} is a set of macrobonds $\calM = \{M_1,\ldots,M_\ell\}$, where each $M_i \subseteq [n]^2$ and $|M_i|=w$, so that for all $1 \leq i < j \leq \ell$, two conditions hold:
\begin{description}[labelindent=0.1cm]
\item[]
{\bf low cross-correlation:} $\forall \vv \in \Z^2$, $M_i \cap (M_j + \vv)| \leq \lambda$.
\item[]
{\bf low auto-correlation:} $\forall \vv \in \Z^2 \setminus \{\vec{0}\}$, $|M_i \cap (M_i + \vv)| \leq \lambda$.
\end{description}

\begin{figure}[ht]
  \begin{centering}
    \includegraphics[width=3.5in]{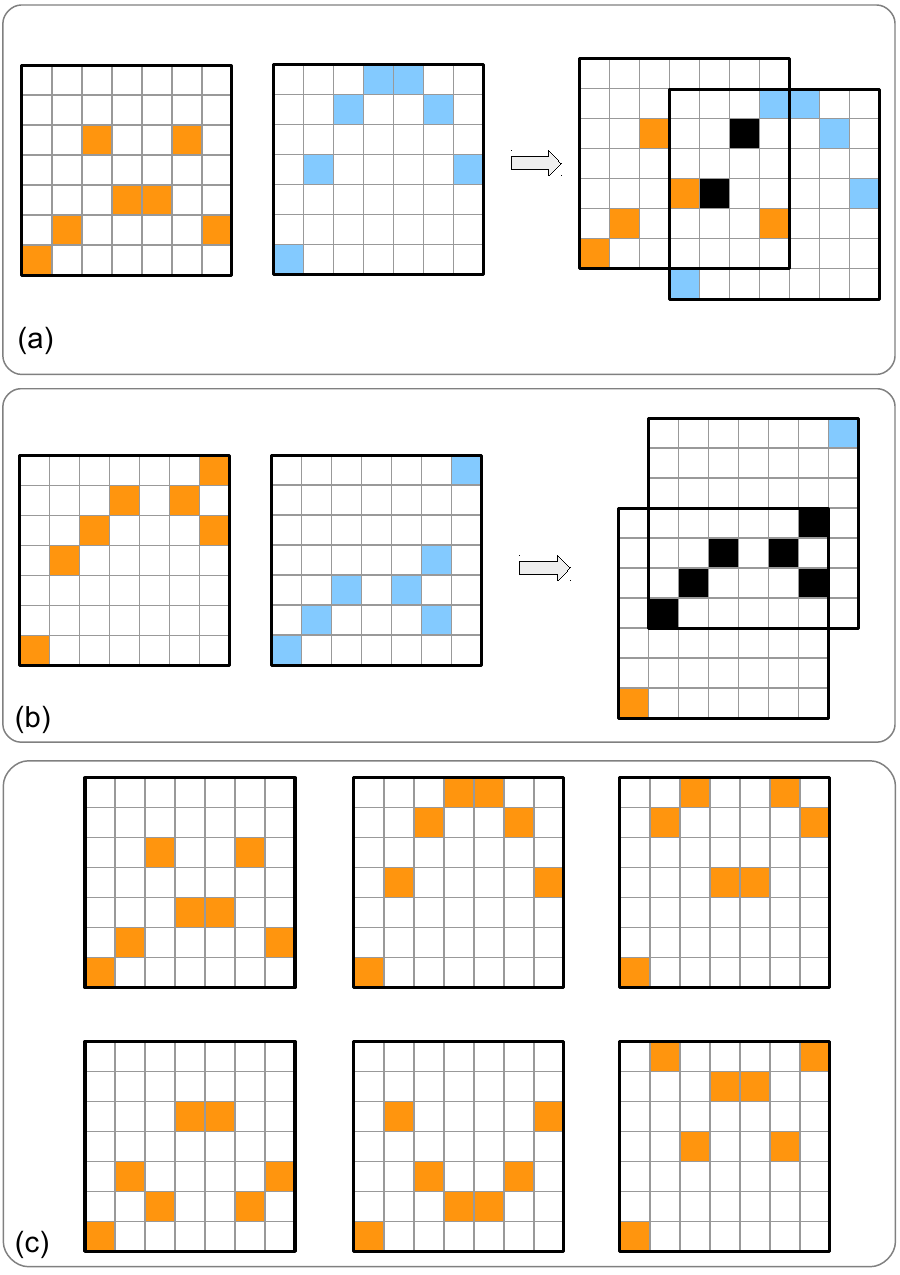}
    \caption{
      {\bf (a)}
      Two macrobonds with cross-correlation $=2$: they have translations overlapping on up to 2 points, but no more.
      {\bf (b)}
      Two macrobonds with cross-correlation $=6$.
      {\bf (c)}
      A $(7,7,2)$-geometric orthogonal code.
    }
    \label{fig:goc-example}
  \end{centering}
\end{figure}

Figure~\ref{fig:goc-example}(a) shows an example of two macrobonds with cross-correlation 2
and a translation that makes them overlap on 2 points.
Figure~\ref{fig:goc-example}(b) shows two macrobonds with cross-correlation 6.
Figure~(c) shows an example of a $(7,7,2)$-geometric orthogonal code.
(This is in fact the code produced by the algorithm of Theorem~\ref{thm:polynomial-construction} with $n=w=7$ and $\lambda=2$.)

Informally, an $(n, w, \lambda)$-geometric  orthogonal code is a set of macrobonds with $n^2$ available space for potential binding sites, total binding strength per macrobond of $w$, and spurious binding strength
(i.e., the strength of mismatched/misaligned macrobonds)
limited to at most $\lambda$.
Note that a lower value of $\lambda$ in the above definition corresponds to ``more orthogonality'':
macrobonds with maximum overlap $\lambda$ are less likely to bind spuriously than those with overlap $> \lambda$.
As with any code, the goal is to maximize the number of codewords $|\mathcal{M}|$.
The main result of this paper, Theorem~\ref{thm:polynomial-construction}, is that for all $n,\lambda \in \Z^+$ with $n$ a prime and $2 \leq \lambda < n$, there exists an efficiently computable $(n, n, \lambda)$-geometric  orthogonal code $\mathcal{M}$ with $|\mathcal{M}| = n^{\lambda-1} - n^{\lambda-2}$.
(Open question~\eqref{open-question:other-w} in Section~\ref{sec:open-questions} discusses the possibility of $(n,w,\lambda)$ codes for $n \neq w$.)

Examine the physical implementation of patches shown in Figure~\ref{fig:origami}(e), and observe that a ``bump'' patch on one macrobond cannot insert into a ``hole'' patch on another macrobond if they are rotated relative to each other,
\emph{unless} the rotation is by 180$^\circ$.\footnote{There are physical reasons to dismiss this possibility, since the energy of a stacking bond appears to weaken if the relative angles of the two DNA phosphate backbones are rotated 180$^\circ$ relative to each other: see the image on the right of Figure~\ref{fig:origami}(c), where only two of four possible orientations appear. The other two would put the $L$ pattern \emph{underneath} the origami; presumably they are absent due to weaker stacking energy of rotated helices. Nonetheless, if one imagines a mixture of different angles being used in different patches, rather than only one angle as in~\cite{sungwook_bonds_2011}, then it may be reasonable to assume a worst-case scenario in which a 180$^\circ$ rotation brings patches into contact such that they bind with the same strength as they would without rotation.}
(Otherwise the blunt ends will not be flush.)
Note that rotating a macrobond by 180$^\circ$ is equivalent to reflecting along each axis once.
To model this scenario, we define an \emph{$(n,w,\lambda)$-geometric 180-rotating orthogonal code} to be an $(n,w,\lambda)$-geometric orthogonal code that, defining $\flip(M_i) = \setr{(n-1-x,n-1-y)}{(x,y)\in M_i},$ also obeys $|\flip(M_i) \cap (M_j + \vv)| \leq \lambda$ for all $1 \leq i \leq j \leq \ell$ and all $\vec{v} \in \Z^2$.
Theorem~\ref{thm:polynomial-construction-flipping} shows codes that obey this extra constraint as well.

\subsection{Relationship between definitions and experimental reality}
\label{sec:intro:relationship-defn-experiments}

The most typical purpose of a code is to enable transmission of data robustly to error,
by ensuring each codeword is a large ``distance'' from all others.
However, the codes we define do not have the purpose of transmitting information.
Rather, the purpose of our codes is to enforce \emph{orthogonality of binding}:
ensuring large ``distance'' between two macrobonds
implies a small amount of binding strength between them,
which is desirable when specificity of binding is the goal.

One may ask, what is the point of designing such specific macrobonds?
One of the unifying goals of structural DNA nanotechnology, since its inception by Seeman in 1982~\cite{Seem82},
has been to \emph{build structures}, specifically to build nanoscale structures out of DNA.
The vast majority of experimental work in this area is bottom-up rather than top-down.
Specifically, one builds \emph{parts} of a structure (the parts being DNA strands, created using standard chemical synthesis approaches~\cite{beaucage1993oligodeoxyribonucleotides}), such that by mixing these parts together, they \emph{autonomously self-assemble} into a larger structure.

In a top-down approach, there is less need for orthogonal bonds.
For example, a carpenter building a house from wood boards (monomers) and nails (bonds)
can use the same type of nail to hold together all pairs of boards,
by making top-down choices to co-locate a pair of boards before driving the nail through.
On the other hand, bottom-up self-assembly is akin to throwing all the boards and nails together into the construction site and hoping they stick together as intended to form a house.
The nails must be designed to stick only to their intended boards;
\emph{orthogonality} means that they have low probability to erroneously hold together other unintended boards.
If more orthogonal bonds are available, then larger and more complex structures can be made.

The specific Watson-Crick base-pairing of DNA is an extremely useful tool for engineering orthogonality of bonds.
However, DNA sticky ends as the sole design tool has limits;
see~\cite{sungwook_bonds_2011} for a more detailed discussion of the relative advantages and disadvantages of DNA sticky end design compared to geometric bond design of the sort we study here.

A simplifying assumption of our definition of $(n,w,\lambda)$ geometric orthogonal codes is that each macrobond has the same number of patches $w$.
If each patch has the same binding strength,
so that the strength with which a macrobond binds is proportional to $w$,
this implicitly assumes that one would want all macrobonds in a system to be the same strength.

This may not be the case in all circumstances.
For example, certain work on algorithmic self-assembly~\cite{Winfree98simulationsof, Rot00, rothemund2004algorithmic}
(for background and surveys of the field, see~\cite{DotCACM, PatitzSurveyJournal, winslow2016brief})
requires that some ``strong'' bonds are twice as strong as other ``weak'' bonds
in order for the desired growth order of molecules to be the most kinetically favorable one.
It is possible to set certain experimental conditions
(e.g., temperature, salinity, concentrations)
to be such that a single weak bond is unfavorable and detaches relatively quickly,
yet two cooperating weak bonds have the strength of a single strong bond
and suffice to attach a molecule stably to a complex of other molecules.
Open Question~\eqref{open-question:multiple-w} in Section~\ref{sec:open-questions} discusses this idea in more detail.

Finally, we note one important distinction between the implementation of patches in Figure~\ref{fig:origami}(d) versus~\ref{fig:origami}(e).
Figure~\ref{fig:origami}(d) uses ``like-like'' binding: an active patch on the edge of one origami resembles the active patch to which it is intended to bind on the edge of another origami.
In contrast, Figure~\ref{fig:origami}(e) uses ``like-unlike'' (a.k.a. \emph{complementary}) binding: a patch implemented as a ``bump'' on one face of an origami is intended to bind to a ``dent'' on the face of another origami.
(Also, bumps on two faces of the first type could align some blunt ends, but this is a source of error not modeled in this paper.)
In modeling like-like binding, it would make sense to consider a macrobond reflecting along the horizontal or vertical axis and coming into contact with another macrobond that has not been reflected (including an unreflected copy of itself).
Since we do not consider this scenario, our definition implicitly assumes like-unlike binding.
Open Question~\eqref{open-question:reflect-one-axis} asks for codes that account for the other scenario.

\subsection{Related work}
\label{sec:related}

The most directly related theoretical work is the study of binary \emph{optical orthogonal codes} defined by Chung, Salehi, and Wei~\cite{optical_orthogonal_codes_1989}.
These codes contain 1D binary codewords and attempt to minimize the number of overlapping 1's (analogous to our nonspecific patches) between codewords;
overlapping 0's (analogous to neutral non-binding sites) are not penalized.
Also, these codes consider all possible translations of codewords; a codeword requires orthogonality not only to translations of other codewords (cross-correlation) but also to nonzero translations of itself (auto-correlation).

The major difference between optical orthogonal codes and our work is the geometric nature of our codes.
Each codeword represents a 2D face of a 3D molecular structure, so translations in both $x$ and $y$ coordinates must be considered.
They also use different parameters to bound auto-correlation and cross-correlation, but for the setting we are modeling, these both correspond to spurious molecular bonds, so it makes sense to use the same threshold for each.

Another difference with our setting is that optical orthogonal codes are more stringent in defining orthogonality under translation, since they use a different definition of translation that allows for more potential overlaps.
In\cite{optical_orthogonal_codes_1989}, translations are assumed to be modulo the codeword size, whereas in our setting such ``wrapping'' does not make sense: a molecular structure $\alpha$ moving off the end of another structure $\beta$ does not appear on the opposite side of $\beta$, hence could not contribute to the binding strength.
See Figure~\ref{fig:ooc}(b).

\begin{figure}[ht]
  \begin{centering}
    \includegraphics[width=3.5in]{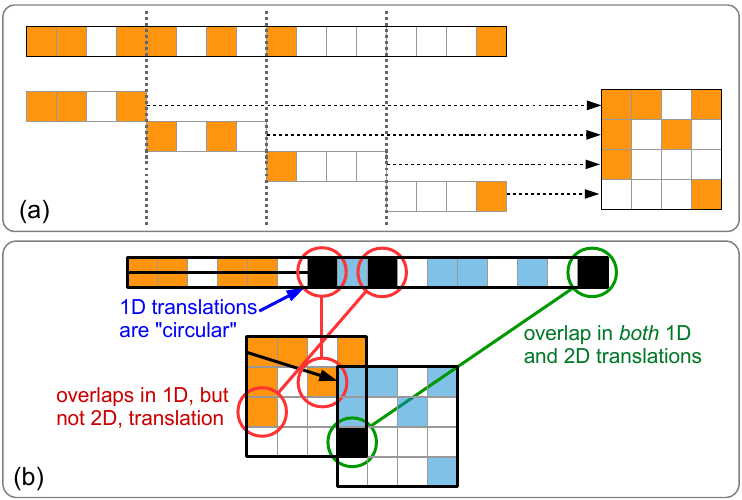}
    \caption{
      {\bf (a)}
      Each 1D codeword of an optical orthogonal codeword of length $n^2$ can be interpreted as a 2D macrobond (a codeword of a geometric orthogonal code) in a natural way by interpreting blocks of length $n$ as rows of the macrobond.
      {\bf (b)}
      Each translation of a 1D codeword by $k\in\Z^+$ can be interpreted as a translation of the equivalent 2D macrobond by $(x,y) = (k \mod n, \lfloor k/n \rfloor).$
      An example 1D translation by $+7$ is shown in blue, which corresponds to translating the 2D codeword by $(+3,-1)$.
      Each overlap between the original (orange) and the translated (blue) 2D macrobonds corresponds to an overlap between the original and translated 1D codewords.
      (See overlaps circled in green.)
      However, the converse does not hold: an overlap in the 1D codewords is not necessarily an overlap between the 2D codewords.
      (See red circles.)
      Also, in an optical orthogonal code, translations are ``circular'': each patch that moves off the end of the codeword wraps back to the beginning, resulting in further potential overlaps in the 1D case that are not counted in the 2D case.
      (See blue arrow.)
      Thus, each $(n^2,w,\lambda)$-optical orthogonal code, interpreting each codeword as a 2D macrobond as in (a), is a $(n,w,\lambda)$-geometric orthogonal code, but the converse does not hold.
    }
    \label{fig:ooc}
  \end{centering}
\end{figure}

Despite these differences, one could imagine applying the optical orthogonal codes of~\cite{optical_orthogonal_codes_1989} directly to our problem setting.
Indeed, every $(n^2,w,\lambda)$-optical orthogonal code is in fact an $(n,w,\lambda)$-geometric orthogonal code,
by interpreting each 1D codeword as the concatenation of the $n$ rows of a 2D codeword.
Figure~\ref{fig:ooc} shows how this interpretation works, and why the two types of codes are not equivalent.
Specifically, the constraints of 1D optical orthogonal codes are stronger than what is needed for 2D geometric orthogonal codes,
partly due to the ``wrapping'' in the definition of codeword translation in 1D optical orthogonal codes.

Table~\ref{tab:upper-lower-comparison} compares the $(n, n, \lambda)$-geometric orthogonal code sizes of our main construction to the $(n^2, n, \lambda)$-optical orthogonal code sizes of the best construction (Theorem 2) of~\cite{optical_orthogonal_codes_1989}, showing that we achieve larger code sizes in most tested cases.
There has been subsequent work on optical orthogonal codes.
However, much of it is for the special case of $\lambda=1$ and/or $w = 4$~\cite{fuji2000optical, yin1998some,chang2003combinatorial,ge2001constructions},~\cite[Section V.9]{colbourn2006handbook} or other special cases for single values of parameters~\cite{chung1990optical, chang2004further, buratti2011new, wang2012further, buratti2013optimal}.
We do note that for the 1D macrobonds studied by Woo and Rothemund~\cite{sungwook_bonds_2011}, 1D optical orthogonal codes are a more appropriate model than the 2D codes that we study in this paper. However, observe that the ``circular translation'' in the model of 1D optical orthogonal codes is not applicable to 1D macrobonds.
Thus, there may be better 1D codes that take advantage of the fact that patches in a macrobond,
translated off the end of a second macrobond,
cannot possible overlap patches in the latter macrobond.

2D optical orthogonal codes have been studied~\cite{twoDimOrthCodeLightwave2005, twoDimOrthCodeIEEETransOnCommunication2011, twoDimOrthCodeArxiv2013}.
The 2D nature of these codes reflects the fact that two different variables (e.g., time and wavelength) determine where 1's and 0's appear in the codeword.
However, these techniques do not apply directly to our problem, since they consider only translation in one dimension (time) while we must consider simultaneous translations along both dimensions.
In other words, the distinction of having two \emph{identical}, \emph{spatial} dimensions is important.

Huntley, Murugan, and Brenner~\cite{InformationCapacityofSpecificInteractionsPNAS} have also studied specific engineered molecular bonds from an information theory perspective.
They study a different model in which translation is disallowed.
They study ``color'' coding: extending the patches to allow some specificity, so that only equal-color patches can bind; see Section~\ref{sec:open-questions} for a discussion of this issue.
They compare color coding with ``shape'' coding: allowing the shape of a 1D edge to be nonflat, thus providing steric hindrance as an additional mechanism to prevent unintended binds (also discussed as an open question in Section~\ref{sec:open-questions}).
They run simulations to show that randomly selected shape codes have greater size than randomly selected color codes.

Since the beginning of DNA nanotechnology~\cite{Seem82} and DNA computing~\cite{Adleman1994}, there has been work on designing codes for DNA-based computers~\cite{frutos1997demonstration, brenneman2002strand, marathe2001oncombinatorial, garzon2004codeword}.
Often these have similar goals to those of the present paper:
designing DNA sequence pairs to bind to each other,
while minimizing unintended binding among all other pairs.
The main difference with our work is not goals, but analysis and techniques.

The biophysics governing the binding of DNA sequences is quite specific to nucleic acids.
Unlike in our model, there is a specificity of binding due to Watson-Crick base pairing: A binds to T and C binds to G.
Another major difference is that nucleic acid sequences are ``flexible'' 1D sequences embedded in 3D:
they can bend, knot, and otherwise contort to potentially allow many pairs of nucleotide bonds to form that are not necessarily the same “distance” apart on the 1D sequence.
On the other hand, in our more ``rigid'' model,
if two patches on a macrobond are separated by a vector $\vec{v}$,
then they can bind to another pair of patches on another macrobond
only if the latter pair are also separated by exactly the vector $\vec{v}$.
In summary, although the ultimate goal is similar
(the design of specific molecular bonds with minimal crosstalk),
the assumptions and techniques are quite different.

\section{Results}
\label{sec:main}

\newcommand{\GF}[1]{\mathbb{F}_{#1}}

\subsection{Lower bounds}

Let $\GF{n}$ denote the finite field of order $n$, where $n$ is prime, which can be interpreted as normal integer addition and multiplication modulo $n$, with field elements $\GF{n}=[n]$. 

The following theorem shows how to construct geometric orthogonal codes \emph{without} the 180-rotation constraint.

\begin{theorem}
\label{thm:polynomial-construction}
For each odd prime $n\in\N$ and $\lambda \in \{2,\ldots,n-1\}$, there is an $(n, n, \lambda)$-geometric orthogonal code of size $n^{\lambda-1} - n^{\lambda-2}$.
\end{theorem}

\begin{proof}
\noindent
{\bf Construction.}
Each macrobond is defined by a degree-$\lambda$ polynomial $p(x) = a_\lambda x^\lambda + a_{\lambda-1} x^{\lambda-1} + \dots + a_1 x + a_0$ over $\GF{n}$,
where the coefficients $a_i \in \GF{n}$ for $i \in \{0,\ldots,\lambda\}$ obey
$a_{\lambda-1} = a_0 = 0$,
and $a_\lambda \neq 0$.
For a polynomial $p$,
define the corresponding macrobond $M_{p} = \setr{ (x,p(x)) }{x \in \GF{n}}$, i.e., a patch in column $x$ on row $p(x)$.
There are $n-1$ choices for $a_\lambda$ and
$n^{\lambda-2}$ choices for $a_{\lambda-2}, a_{\lambda-3}, \dots, a_1$,
so there are $(n-1) n^{\lambda-2} = n^{\lambda-1}-n^{\lambda-2}$ such polynomials.

\noindent
{\bf Correctness.}
We now show that this code has auto-correlation and cross-correlation at most $\lambda$.
Translation by a vector $\vv = (\delta_x, \delta_y)\in\Z^2$ with $|\delta_x|$ or $|\delta_y|$ $\geq n$ implies correlation is~0.
So assume $|\delta_x|, |\delta_y| < n$.
For mathematical convenience we equivalently consider translating by $(-\delta_x, \delta_y)$.
Suppose there exist macrobonds $M_p$ and $M_q$ (possibly equal) with $p(x) = \sum_{i=0}^\lambda{a_i x^i}$ and $q(x) = \sum_{i=0}^\lambda{b_ix^i}$ that intersect on more than $\lambda$ patches under translation $(-\delta_x, \delta_y)$.
We will prove that $a_i=b_i$ for all $i\in\{0,\ldots,\lambda\}$ (i.e., the two macrobonds are the same) and that $\delta_x=\delta_y=0$ (i.e., the translation is $\vec{0}$), simultaneously establishing that the code has auto-correlation and cross-correlation at most $\lambda$.

Translating the polynomial $p$ by $(-\delta_x, \delta_y)$ results in the polynomial $p(x+\delta_x) + \delta_y$.
If this intersects with the polynomial $q(x)$ on more than $\lambda$ points,
then by the fundamental theorem of algebra,
$p(x+\delta_x) + \delta_y$ is identically $q(x)$,
i.e., they have the same coefficients.
Using the binomial theorem,
\begin{align*}
  p(x+\delta_x) +\delta_y
  = & \
  \delta_y + \sum_{i=0}^\lambda a_i (x+\delta_x)^i
  \\ =
  \delta_y + \sum_{i=0}^\lambda a_i \sum_{k=0}^i \binom{i}{k} \delta_x^{i-k} x^k
  = & \
  \delta_y + \sum_{k=0}^\lambda x^k \sum_{i=k}^\lambda a_i \binom{i}{k} \delta_x^{i-k}
\end{align*}
Thus, if $p(x+\delta_x)+\delta_y$ and $q(x)$ have the same coefficients,
then the $q(x)$ coefficient of the term $x^k$ is
$
  b_k = \sum_{i=k}^\lambda a_i \binom{i}{k} \delta_x^{i-k}.
$

In particular
$
b_{\lambda-1}
= \sum_{i=\lambda-1}^\lambda a_i \binom{i}{\lambda-1} \delta_x^{i-(\lambda-1)}
= a_{\lambda-1} \binom{\lambda-1}{\lambda-1} \delta_x^{\lambda-1-(\lambda-1)}
  +
  a_{\lambda} \binom{\lambda}{\lambda-1} \delta_x^{\lambda-(\lambda-1)}
= a_{\lambda-1} + a_\lambda \lambda \delta_x.
$
Since $a_{\lambda-1}=b_{\lambda-1}=0$, this implies $a_\lambda \lambda \delta_x=0$.
Since $\lambda,a_\lambda \neq 0$, this implies $\delta_x = 0.$

The constant term of $p(x+\delta_x) + \delta_y = p(x)+\delta_y$ is $a_0+\delta_y$,
and the constant term of $q(x)$ is $b_0$.
Since the coefficients of $p$ and $q$ are equal, $a_0+\delta_y = b_0$, but since $a_0=b_0=0$, this implies $\delta_y=0$ also.
\end{proof}

Note that the macrobonds in the construction of Theorem~\ref{thm:polynomial-construction} are ``column-balanced'': there is exactly one patch per column of the macrobond.
This is due to our proof technique and is not itself a goal of the macrobonds.
Such column-balanced macrobonds are a bit easier to reason about theoretically,
so they come up again in Theorem~\ref{thm:random-construction} when we prove bounds on the size of random codes.
However, we know of no intrinsic benefit to the column-balanced property.

We now construct a geometric \emph{180-rotating} orthogonal code, using similar techniques to the previous proof.

\begin{theorem}
\label{thm:polynomial-construction-flipping}
For each odd prime $n\in\N$ and $\lambda \in \{2,\ldots,n-1\}$, there is an $(n, n, \lambda)$-geometric 180-rotating orthogonal code of size $(n^{\lambda-1} - n^{\lambda-2} - n^{\lceil \lambda/2 \rceil})/2$.
\end{theorem}

\begin{proof}
The construction is a modification of Theorem~\ref{thm:polynomial-construction} obtained by taking a subset of the code that avoids high correlation in the new orientation.
Assume that there exist polynomials $p$ and $q$
and a translation vector $(\delta_x,\delta_y)$
such that $M_q$ has correlation $> \lambda$ with $\flip(M_p)+(\delta_x,\delta_y)$.
By definition, $\flip(M_p) = \setr{(n-1-x,n-1-p(x))}{x \in \GF{n}} = \setr{(x,n-1-p(n-1-x))}{x \in \GF{n}}$,
so by the fundamental theorem of algebra, $n-1-p(n-1-x + \delta_x) + \delta_y = q(x)$ for all $x\in\GF{n}$.

The $x^{\lambda-1}$ term of $q(x)$ is 0 by construction.
Expanding the two leading terms of $n-1-p(n-1-x + \delta_x) + \delta_y$ with this constraint implies that $a_\lambda \lambda (n-1+\delta_x) (-1)^\lambda = 0$.
Then since $\lambda, a_\lambda, (-1)^\lambda \neq 0$, it must be that $n-1+\delta_x = 0$.
So $\delta_x = 1$ and $n-1-p(n-1-x + \delta_x) + \delta_y = n-1-p(-x) + \delta_y$.

Expanding all terms of $n-1-p(-x) + \delta_y$ and $q(x)$ leads to $n-1 + \delta_y = n-1 + a_0 + \delta_y = b_0 = 0$.
So $\delta_y = 1$ and thus $n-1-p(n-1-x + \delta_x) + \delta_y = n-p(n-x) = -p(-x) = q(x)$.
Expanding these polynomials implies that for all $i$, $a_i(-1)^{i+1} = b_i$.

Define a \emph{complement} of a polynomial $p(x)$ to be $\sum_{i=1}^\lambda{a_i(-1)^{i+1}} x^i$, i.e. the polynomial $q(x)$ such that $-p(-x) = q(x)$.
As the above shows, the current code allows two macrobonds to have correlation $> \lambda$ points if one of them is rotated by 180 degree, but only if the two corresponding polynomials are complements.
Observe that every polynomial has a unique complement, and some polynomials are complements of themselves.
Self-complement polynomials have auto-correlation more than $\lambda$ and complementary pairs have cross-correlation more than $\lambda$.
A 180-rotating code can be obtained by taking any subset of the code that contains no polynomial and its complement.

A self-complementary polynomial is one in which for all even $i$, $a_i + a_i=0$.
Since $n > \lambda \geq 2$ is prime, $n$ is odd and $a_i + a_i = 0$ only if $a_i = 0$.
Thus, the number of self-complementary polynomials is at most $n^{\lceil \lambda/2 \rceil}$.

First remove all such self-complementary polynomials from the code.
The remaining polynomials occur in uniquely complementary pairs; remove one member of each pair arbitrarily, cutting the number of remaining polynomials in half.
So the 180-rotating code has size $(n^{\lambda-1} - n^{\lambda-2} - n^{\lceil \lambda/2 \rceil})/2$.
\end{proof}


The proofs of Theorems~\ref{thm:polynomial-construction} and~\ref{thm:polynomial-construction-flipping} use finite field arithmetic only for fields of prime size, even though there are finite fields of size $p^m$ for any prime $p$ and $m\in\Z^+$.
This is due to our technique of mapping the field elements to the integers $\{0,1,\ldots,n-1\}$ in such a way that translations in the $x$ and $y$ direction can be interpreted as changes in the underlying field elements.
To make the correspondence straightforward, the characteristic of the field
(the number of times the multiplicative identity 1 can be added before a repetition)
must be equal to the field size, which is true exactly when the field size is prime.
For prime size fields, translating $x$ by the integer $m$ is the same as adding 1 to $x$, $m$ times in a row.
Otherwise, translation from a point $(x,y)$ to a point $(x',y')$, where (for example), $x'-x$ is greater than the field's characteristic but less than its size, would not be interpretable as mapping the element $x$ to $x'$ by repeated addition, and would invalidate the parts of the proof that reason about the effects of translation on the underlying polynomial evaluation.

\subsection{Upper bounds}

As observed in Section~\ref{sec:related},
any $(n^2,w,\lambda)$-optical orthogonal code is automatically a $(n,w,\lambda)$-geometric orthogonal code (by re-arranging from 1D to 2D as in Figure~\ref{fig:ooc}),
so lower bounds on the size of $(n^2,w,\lambda)$-optical orthogonal codes also hold for $(n,w,\lambda)$-geometric orthogonal codes.
However, the converse does not hold:
not every $(n,w,\lambda)$-geometric orthogonal code is a $(n^2,w,\lambda)$-optical orthogonal code.
Thus \emph{upper} bounds on the size of $(n^2,w,\lambda)$-optical orthogonal codes
(such as the $n^2(n^2-1) \ldots (n^2-\lambda) / (w (w-1) \ldots (w-\lambda))$ upper bound proved in~\cite{optical_orthogonal_codes_1989})
do not automatically apply to geometric orthogonal codes.
In principle, due to the relaxed constraints of geometric orthogonal codes, their optimal sizes could potentially be larger than for optical orthogonal codes.

The following theorem shows an upper bound on the size of any geometric orthogonal code.
Intuitively, the proof is a packing argument that works as follows.
Given a set $S \subseteq [n]^2$,
imagine a ``canonical'' translation $S^\llcorner$ of $S$
so that it is ``flush'' against the $x$- and $y$-axes:
the translation has at least one $x$-coordinate and at least one $y$-coordinate equal to 0,
but no negative coordinates.
Two sets $S,T \subseteq [n]^2$ obey $|S \cap T| \leq \lambda$ if and only if,
for all subsets $S_\lambda \subseteq S$ and $T_\lambda \subseteq T$
such that $|S_\lambda| = |T_\lambda| = \lambda+1$,
we have $S_\lambda \neq T_\lambda$.
Furthermore, $S_\lambda$ and $T_\lambda$ are equal under some translation if and only if
$S_\lambda^\llcorner = T_\lambda^\llcorner$.
Let $\calM$ be a $(n, w, \lambda)$-geometric orthogonal code.
Each macrobond has precisely $\binom{w}{\lambda+1}$ subsets of size $\lambda+1$,
so across all $|\calM|$ macrobonds, there are $|\calM| \cdot \binom{w}{\lambda+1}$ total induced subsets of size exactly $\lambda+1$.
The code has auto- or cross-correlation more than $\lambda$ if and only if there exists a pair of these subsets having equal canonical translations.
We count the number of distinct canonically translated subsets of $[n]^2$ of size $\lambda+1$, observing that $|\calM| \cdot \binom{w}{\lambda+1}$ must be below this count to avoid repeating a subset by the pigeonhole principle.

\begin{theorem}
\label{thm:upper-bound}
Any $(n, w, \lambda)$-geometric orthogonal code has size at most
$\frac{1}{\binom{w}{\lambda+1}}
\cdot
\left[ \binom{n^2-1}{\lambda}
+ \sum_{x_0=1}^{n-1} \sum_{y_0=1}^{n-1} \binom{n^2 - x_0 - y_0 - 1}{\lambda-1}
\right]$.
\end{theorem}

\begin{proof}
For $S \subseteq [n]^2$, define $S^{\llcorner} = S+(-x_\text{min},-y_\text{min})$, where
$x_\text{min} = \min\limits_{(x,y) \in S}(x)$
and
$y_\text{min} = \min\limits_{(x,y) \in S}(y)$,
to be the \emph{canonical translation} of $S$.
Note that $S^\llcorner \subseteq [n]^2$, and $S^\llcorner$ has at least one point on the $x$-axis and at least one point on the $y$-axis.

Each macrobond $M \subseteq [n]^2$ with $w=|M|$ has exactly $\binom{w}{\lambda+1}$ subsets of size exactly $\lambda+1$.
Denote these subsets as $M_{\lambda, 1},M_{\lambda, 2},\ldots,M_{\lambda,\binom{w}{\lambda+1}}$.
Note that macrobond $M$ has auto-correlation $\leq \lambda$ if and only if
$M^\llcorner_{\lambda,i} \neq M^\llcorner_{\lambda,j}$ for all $1 \leq i < j \leq \binom{w}{\lambda+1}$,
and that macrobonds $M$ and $N$ have cross-correlation $\leq \lambda$ if and only if
$M^\llcorner_{\lambda,i} \neq N^\llcorner_{\lambda,j}$ for all $1 \leq i,j \leq \binom{w}{\lambda+1}$.
To avoid making any two of these translated subsets equal,
an $(n,w,\lambda)$-geometric orthogonal code $\calM$ obeys
$|\calS| = |\calM| \cdot \binom{w}{\lambda+1}$,
where $\calS = \{ M^\llcorner_{\lambda,i} \mid M \in \calM, 1 \leq i \leq \binom{w}{\lambda+1} \}$.
Thus $|\calM| = \frac{1}{\binom{w}{\lambda+1}} |\calS|$, and to prove the theorem, it suffices to show
$|\calS| \leq \binom{n^2-1}{\lambda} + \sum_{x_0=1}^{n-1} \sum_{y_0=1}^{n-1} \binom{n^2 - x_0 - y_0 - 1}{\lambda-1}$.

To bound $|\calS|$, we simply count the number of canonical translations $S^\llcorner$ of subsets $S \subseteq [n]^2$ with $|S|=\lambda+1$.
To be a canonical translation, $S^\llcorner$ must have at least one point on the $x$-axis and at least one point on the $y$-axis.
We count two subcases separately.
First assume $(0,0) \in S^\llcorner$.
Then there are $\binom{n^2-1}{\lambda}$ ways to pick the other $\lambda$ points in $S^\llcorner$ besides $(0,0)$.

Now assume the other case: $(0,0) \not\in S^\llcorner$.
Let $x_0 = \min\limits_{(x,0) \in S^\llcorner} x$ and $y_0 = \min\limits_{(0,y) \in S^\llcorner} y$. That is, $x_0$ and $y_0$ are, respectively, the smallest $x$- and $y$-coordinates of points in $S^\llcorner$ whose other coordinate is 0.
Because $(0,0) \not\in S^\llcorner$, we have $x_0,y_0 > 0$.
Once the two points defining $x_0$ and $y_0$ are fixed, there are $\lambda-1$ other points to pick to be in $S^\llcorner$,
and they must be picked from the set
$A=\{(x,y)\in[n]^2 \mid   (x=0 \implies y > y_0) \wedge (y=0 \implies x > x_0) \}$.
Note that
$|A| = (n-1)^2 + (n-x_0-1) + (n-y_0-1)
= n^2 - x_0 - y_0 - 1$,
where $(n-1)^2$ is the number of available points off both axes,
and the terms $(n-x_0-1)$ and $(n-y_0-1)$ count the number of available points on each axis.
There are thus $\binom{n^2 - x_0 - y_0 - 1}{\lambda-1}$ ways to pick these $(\lambda-1)$ points from $A$.
Therefore there are $\sum_{x_0=1}^{n-1} \sum_{y_0=1}^{n-1} \binom{n^2 - x_0 - y_0 - 1}{\lambda-1}$ total sets in this subcase. 
\end{proof}

\newcommand{\Lcon}{L_{\mathrm{\ref{thm:polynomial-construction}}}}
\newcommand{\Looc}{L_{\mathrm{ooc}}}
\newcommand{\Udet}{U_{\mathrm{\ref{thm:upper-bound}}}}
\newcommand{\Uran}{U_{\mathrm{\ref{thm:random-construction}}}}

\begin{table}[ht]
    \centering
    \caption{
    Comparison of code size lower bound $\Lcon(n,\lambda)$ of Theorem~\ref{thm:polynomial-construction} with code size upper bound $\Udet(n,\lambda)$ (with $w=n$) of Theorem~\ref{thm:upper-bound} and lower bound $\Looc(n^2,\lambda)$ given by optical orthogonal code construction of Theorem 2 of\cite{optical_orthogonal_codes_1989}.
    }
    \label{tab:upper-lower-comparison}
    \begin{tabular}{|r|r|r|r|r|}
      \hline
      $n$ & $\lambda$ & $\Lcon(n,\lambda)$ & $\Udet(n,\lambda)$ & $\Looc(n^2,n,\lambda)$ \\
      \hline \hline
      5 & 2 & 4 & 58 & 0 \\
      5 & 3 & 20 & 956 & 20 \\
      5 & 4 & 100 & 26,490 & 2,124 \\
      \hline
      7 & 2 & 6 & 74 & 0 \\
      7 & 3 & 42 & 1,340 & 3 \\
      7 & 4 & 294 & 27,740 & 94 \\
      7 & 5 & 2,058 & 777,148 & 5,942 \\
      7 & 6 & 14,406 & 40,291,608 & 1,753,072 \\
      \hline
      11 & 2 & 10 & 109 & 0 \\
      11 & 3 & 110 & 2,637 & 0 \\
      11 & 4 & 1,210 & 63,413 & 9 \\
      11 & 5 & 13,310 & 1,626,997 & 179 \\
      11 & 6 & 146,410 & 46,982,678 & 5,435 \\
      \hline
      13 & 2 & 12 & 127 & 0 \\
      13 & 3 & 156 & 3,491 & 0 \\
      13 & 4 & 2,028 & 93,188 & 4 \\
      13 & 5 & 26,364 & 2,564,783 & 76 \\
      13 & 6 & 342,732 & 75,841,707 & 1,690 \\
      \hline
      17 & 2 & 16 & 162 & 0 \\
      17 & 3 & 272 & 5,592 & 0 \\
      17 & 4 & 4,624 & 181,316 & 1 \\
      17 & 5 & 78,608 & 5,850,750 & 24 \\
      17 & 6 & 1,336,336 & 194,074,096 & 389 \\
      \hline
      19 & 2 & 18 & 180 & 0 \\
      19 & 3 & 342 & 6,837 & 0 \\
      19 & 4 & 6,498 & 241,967 & 0 \\
      19 & 5 & 123,462 & 8,434,602 & 15 \\
      19 & 6 & 2,345,778 & 298,556,284 & 234 \\
      \hline
    \end{tabular}
\end{table}

Table~\ref{tab:upper-lower-comparison} compares the code size lower bound
$\Lcon(n,\lambda) = n^{\lambda-1} - n^{\lambda-2}$
achieved by the algorithm of Theorem~\ref{thm:polynomial-construction}
with the upper bound of Theorem~\ref{thm:upper-bound} (for the special case of $w=n$)
$\Udet(n,\lambda) = \frac{1}{\binom{n}{\lambda+1}}  \left[ \binom{n^2-1}{\lambda}  + \sum_{x_0=1}^{n-1} \sum_{y_0=1}^{n-1} \binom{n^2 - x_0 - y_0 - 1}{\lambda-1}\right]$.
Also shown is the lower bound
$\Looc(n^2,n,\lambda) = \frac{\binom{n^2}{n} - \frac{n^2-1}{2} \binom{n}{\lambda+1} \binom{n^2}{n-\lambda-1} }{n^2 \cdot \sum_{i=\lambda+1}^{\min(n^2-n,n)} \binom{n^2-n}{n-i} \binom{n}{i} }$
given by the $(n^2,n,\lambda)$-optical orthogonal code construction of Theorem~2 of~\cite{optical_orthogonal_codes_1989}.

Theorem~\ref{thm:upper-bound} is our strongest upper bound but is unwieldy.
The following corollary gives a weaker but simpler bound.

\begin{corollary}\label{cor:upper-bound-simplified}
Any $(n, w, \lambda)$-geometric orthogonal code has size at most
$\frac{(\lambda+1)^2 e^{\lambda+1}}{w^{\lambda+1}}  n^{2\lambda}$.
\end{corollary}

\begin{proof}
  We use the bounds
  $\binom{m-1}{k} < \binom{m}{k}$,
  $\binom{m-1}{k-1} = \frac{k}{m}\binom{m}{k}$,
  and
  $\frac{m^k}{k^k} < \binom{m}{k} < \frac{e^k \cdot m^k}{k^k}$,  for all $m,k \in \Z^+$.
  Then
  \begin{align*}
    &\binom{n^2-1}{\lambda} + \sum_{x_0=1}^{n-1} \sum_{y_0=1}^{n-1} \binom{n^2 - x_0 - y_0 - 1}{\lambda-1}
    \\ & <
    \binom{n^2}{\lambda} + \sum_{x_0=1}^{n} \sum_{y_0=1}^{n} \binom{n^2-1}{\lambda-1}
    =
    \binom{n^2}{\lambda} + n^2 \binom{n^2-1}{\lambda-1}
    \\ & =
    \binom{n^2}{\lambda} + n^2 \frac{\lambda}{n^2} \binom{n^2}{\lambda}
    = 
    (\lambda+1) \binom{n^2}{\lambda}
    < 
    (\lambda+1) \frac{ e^{\lambda} n^{2\lambda} }{\lambda^\lambda}.
  \end{align*}

  Also, $\binom{w}{\lambda+1} > \frac{w^{\lambda+1}}{(\lambda+1)^{\lambda+1}}.$
  Combining these bounds with Theorem~\ref{thm:upper-bound}, we have that the size of any $(n,w,\lambda)$-geometric orthogonal code is at most
  \begin{align*}
  &
  \frac{1}{\binom{w}{\lambda+1}}
  \cdot
  \left( \binom{n^2-1}{\lambda}
+ \sum_{x_0=1}^{n-1} \sum_{y_0=1}^{n-1} \binom{n^2 - x_0 - y_0 - 1}{\lambda-1}\right)
  \\ & <
  \frac{1}{\frac{w^{\lambda+1}}{(\lambda+1)^{\lambda+1}}}
  \cdot
  (\lambda+1) \frac{ e^{\lambda} n^{2\lambda} }{\lambda^\lambda}
  =
  \frac{ (\lambda+1) e^{\lambda} }{w^{\lambda+1}}
  \cdot
  \frac{ (\lambda+1)^{\lambda+1}}{\lambda^{\lambda}}
  \cdot
  n^{2\lambda}
  \\ & =
  \frac{ (\lambda+1)^2 e^{\lambda} }{w^{\lambda+1}}
  \cdot
  \left( \frac{\lambda+1}{\lambda} \right)^{\lambda}
  \cdot
  n^{2\lambda}
  \\ & <
  \frac{ (\lambda+1)^2 e^{\lambda} }{w^{\lambda+1}}
  \cdot
  e
  \cdot
  n^{2\lambda}
  =
  \frac{(\lambda+1)^2 e^{\lambda+1}}{w^{\lambda+1}}
  \cdot
  n^{2\lambda}.  \qedhere
  \end{align*}
\end{proof}

The next corollary applies in the special case where $n=w$.

\begin{corollary}\label{cor:upper-bound-n-equals-w}
Any $(n, n, \lambda)$-geometric orthogonal code has size at most
$(\lambda+1)^2 e^{\lambda+1} n^{\lambda-1}$.
\end{corollary}

Note that the upper bound of Corollary~\ref{cor:upper-bound-n-equals-w} asymptotically matches the lower bound of Theorem~\ref{thm:polynomial-construction} when $\lambda$ is constant with respect to $n$.

\subsection{Random codes}

Although simple and efficient, it is worth asking if the technique of Theorem~\ref{thm:polynomial-construction} is overkill, compared to the most obvious attempt to generate codes: picking macrobonds at random.
The next theorem shows that this approach yields much smaller codes if required to have one patch per column.

\begin{theorem}
\label{thm:random-construction}
    Let $0 < \epsilon < 1$.
    Let $\calM$ be a set of macrobonds selected uniformly at random with replacement
    from among those macrobonds with exactly one patch per column.
    If $|\calM| \geq \frac{\lambda+1}{n} \left(1 + \sqrt{2n^{(\lambda+1)} \ln \frac{1}{\epsilon}} \right)$,
    then $\calM$ is a $(n, n, \lambda)$-geometric orthogonal code with probability at most $\epsilon$.
\end{theorem}

\begin{proof}
The described macrobond distribution can be generated by iterating over each column and selecting one row uniformly at random to contain the patch in that column.
Each selection of patches in nonoverlapping \emph{blocks} of $\lambda+1$ consecutive columns can be viewed as a symbol in an alphabet of size $n^{\lambda+1}$.
Because the blocks are nonoverlapping, each symbol selection is independent.
Then each macrobond is partially specified by $\lfloor n / (\lambda + 1) \rfloor$ symbols defining the patch placements in the first $(\lambda + 1) \lfloor n / (\lambda + 1) \rfloor$ columns.
If any symbol is repeated (either within a macrobond, or between two different macrobonds), then the code has auto- or cross-correlation $> \lambda$.

The probability that $k$ randomly selected symbols from an alphabet of size $n^{\lambda+1}$ does not repeat a symbol is
\begin{align*}
    \prod_{i=0}^{k-1}(1 - i/n^{\lambda+1})
    & <
    \prod_{i=0}^{k-1}{e^{-i/n^{\lambda+1}}}
    = 
    e^{-\sum_{i=0}^{k-1} i / n^{\lambda+1}}
    \\ &=
    e^{-k(k-1)/(2n^{\lambda+1})}
    < 
    e^{-(k-1)^2/(2n^{\lambda+1})}.
\end{align*}
Thus the probability that a symbol \emph{does} repeat is at least $1 - e^{-(k-1)^2/(2n^{\lambda+1})}$.
By algebra, the inequality $1-\epsilon \leq 1 - e^{-(k-1)^2/(2n^{\lambda+1})}$ holds provided $k \geq 1 + \sqrt{2n^{\lambda+1} \ln \frac{1}{\epsilon}}$.

Since each macrobond induces
$\lfloor n / (\lambda + 1) \rfloor$
symbols, a set of
$\frac{\lambda+1}{n} \left( 1 + \sqrt{2n^{(\lambda+1)} \ln\frac{1}{\epsilon}} \right)$
macrobonds induces
$\frac{\lambda+1}{n} \left( 1 + \sqrt{2n^{(\lambda+1)} \ln \frac{1}{\epsilon}} \right) \cdot \left\lfloor \frac{n}{\lambda + 1} \right\rfloor \geq 1 + \sqrt{2n^{(\lambda+1)} \ln \frac{1}{\epsilon}}$ symbols and thus with probability $\geq 1 - \epsilon$ contains a repetition.
\end{proof}

\begin{table}[ht]
    \centering
    \caption{
    Empirical test of random codes for $w=n$.
    }
    \label{tab:random-tests}
    \begin{tabular}{|r@{\hspace{2pt}}|r@{\hspace{2pt}}|r@{\hspace{2pt}}|r@{\hspace{2pt}}|r@{\hspace{2pt}}|r@{\hspace{2pt}}|r@{\hspace{2pt}}|r@{\hspace{2pt}}|r@{\hspace{2pt}}|}
      \hline
      $n$ & $\lambda$ & ave & med & stddev & max & $\Lcon(n,\lambda)$ & $\Uran(n,\lambda,\frac{1}{2})$ \\
      \hline \hline
      5 & 2 & 2 & 2 & 1.2 & 6 & 4 & 8 \\
      5 & 3 & 14 & 13 & 7.5 & 41 & 20 & 24 \\
      5 & 4 & 160 & 151 & 76.7 & 378 & 100 & 66 \\
      \hline
      7 & 2 & 2 & 1 & 0.7 & 4 & 6 & 9 \\
      7 & 3 & 7 & 6 & 3.7 & 17 & 42 & 33 \\
      7 & 4 & 33 & 30 & 18.0 & 96 & 294 & 109 \\
      7 & 5 & 303 & 282 & 182.0 & 972 & 2,058 & 347 \\
      \hline
      11 & 2 & 1 & 1 & 0.3 & 2 & 10 & 11 \\
      11 & 3 & 3 & 3 & 1.5 & 9 & 110 & 52 \\
      11 & 4 & 12 & 12 & 5.7 & 30 & 1,210 & 215 \\
      11 & 5 & 59 & 59 & 28.2 & 125 & 13,310 & 855 \\
      11 & 6 & 369 & 333 & 216.7 & 1038 & 146,410 & 3,308 \\
      \hline
      13 & 2 & 1 & 1 & 0.1 & 2 & 12 & 12 \\
      13 & 3 & 2 & 2 & 1.2 & 6 & 156 & 61 \\
      13 & 4 & 9 & 9 & 4.9 & 22 & 2,028 & 276 \\
      13 & 5 & 37 & 35 & 20.7 & 111 & 26,364 & 1,194 \\
      13 & 6 & 213 & 208 & 99.6 & 469 & 342,732 & 5,022 \\
      \hline
      17 & 2 & 1 & 1 & 0.0 & 1 & 16 & 14 \\
      17 & 3 & 2 & 1 & 0.6 & 4 & 272 & 80 \\
      17 & 4 & 5 & 4 & 3.1 & 14 & 4,624 & 412 \\
      17 & 5 & 22 & 20 & 12.8 & 70 & 78,608 & 2,042 \\
      17 & 6 & 99 & 91 & 50.7 & 219 & 1,336,336 & 9,821 \\
      \hline
      19 & 2 & 1 & 1 & 0.0 & 1 & 18 & 15 \\
      19 & 3 & 1 & 1 & 0.6 & 4 & 342 & 89 \\
      19 & 4 & 4 & 4 & 2.2 & 11 & 6,498 & 487 \\
      19 & 5 & 17 & 16 & 8.5 & 40 & 123,462 & 2,550 \\
      19 & 6 & 82 & 79 & 41.1 & 175 & 2,345,778 & 12,969 \\
      \hline
    \end{tabular}
\end{table}

Table~\ref{tab:random-tests} shows the result of testing random codes for $n=w$ and several prime values of $n$, comparing them to the proved theoretical bounds of Theorems~\ref{thm:polynomial-construction} and~\ref{thm:random-construction}.
In each row of Table~\ref{tab:random-tests}, 100 trials were run.
In each trial, macrobonds were generated by selecting $n$ patches uniformly at random (without replacement) from the $n \times n$ square.
Macrobonds were generated successively and added to the code until the auto- or cross-correlation of the code exceeded $\lambda$, and the code size recorded for the trial.
The average (``ave'', rounded to nearest integer), median (``med''), standard deviation (``stddev''), and maximum (``max'') code sizes among the 100 trials are shown.
These are compared to the lower bound $\Lcon(n,\lambda) = n^{\lambda-1} - n^{\lambda-2}$ achieved by the algorithm of Theorem~\ref{thm:polynomial-construction},
as well as the randomized upper bound
$\Uran(n,\lambda,\frac{1}{2}) = \frac{\lambda+1}{n} \left(1 + \sqrt{2n^{(\lambda+1)} \ln 2} \right)$
proven in Theorem~\ref{thm:random-construction} for random codes restricted to one patch per column (setting
$\epsilon=\frac{1}{2}$).

\begin{table}[ht]
    \centering
    \caption{
    Empirical test of random codes with one patch per column.
    }
    \label{tab:random-tests-one-patch-per-column}
    \begin{tabular}{|r@{\hspace{2pt}}|r@{\hspace{2pt}}|r@{\hspace{2pt}}|r@{\hspace{2pt}}|r@{\hspace{2pt}}|r@{\hspace{2pt}}|r@{\hspace{2pt}}|r@{\hspace{2pt}}|r@{\hspace{2pt}}|}
      \hline
      $n$ & $\lambda$ & ave & med & stddev & max & $\Lcon(n,\lambda)$ & $\Uran(n,\lambda,\frac{1}{2})$ \\
      \hline \hline
      5 & 2 & 2 & 2 & 1.1 & 6 & 4 & 8 \\
      5 & 3 & 8 & 8 & 4.3 & 19 & 20 & 24 \\
      5 & 4 & 51 & 48 & 25.5 & 119 & 100 & 66 \\
      \hline
      7 & 2 & 1 & 1 & 0.5 & 3 & 6 & 9 \\
      7 & 3 & 4 & 4 & 2.3 & 11 & 42 & 33 \\
      7 & 4 & 20 & 17 & 10.4 & 64 & 294 & 109 \\
      7 & 5 & 92 & 86 & 52.8 & 248 & 2,058 & 347 \\
      \hline
      11 & 2 & 1 & 1 & 0.1 & 2 & 10 & 11 \\
      11 & 3 & 2 & 2 & 1.2 & 8 & 110 & 52 \\
      11 & 4 & 8 & 7 & 4.1 & 21 & 1,210 & 215 \\
      11 & 5 & 32 & 29 & 16.6 & 77 & 13,310 & 855 \\
      11 & 6 & 143 & 131 & 80.0 & 497 & 146,410 & 3,308 \\
      \hline
      13 & 2 & 1 & 1 & 0.0 & 1 & 12 & 12 \\
      13 & 3 & 2 & 2 & 0.9 & 5 & 156 & 61 \\
      13 & 4 & 7 & 6 & 3.8 & 27 & 2,028 & 276 \\
      13 & 5 & 22 & 19 & 11.9 & 62 & 26,364 & 1,194 \\
      13 & 6 & 108 & 92 & 55.4 & 300 & 342,732 & 5,022 \\
      \hline
      17 & 2 & 1 & 1 & 0.0 & 1 & 16 & 14 \\
      17 & 3 & 1 & 1 & 0.6 & 4 & 272 & 80 \\
      17 & 4 & 4 & 4 & 2.4 & 13 & 4,624 & 412 \\
      17 & 5 & 14 & 13 & 6.7 & 31 & 78,608 & 2,041 \\
      17 & 6 & 57 & 54 & 29.2 & 125 & 1,336,336 & 9,821 \\
      \hline
      19 & 2 & 1 & 1 & 0.0 & 1 & 18 & 15 \\
      19 & 3 & 1 & 1 & 0.4 & 2 & 342 & 89 \\
      19 & 4 & 3 & 3 & 1.7 & 12 & 6,498 & 487 \\
      19 & 5 & 13 & 12 & 6.5 & 40 & 123,462 & 2,550 \\
      19 & 6 & 51 & 47 & 26.7 & 118 & 2,345,778 & 12,969 \\
      \hline
    \end{tabular}
\end{table}

Table~\ref{tab:random-tests-one-patch-per-column} shows test results for random codes restricted to exactly one patch per column.
In this case, each random macrobond is generated by selecting, in each of $n$ columns, one row uniformly at random in which to place a patch.

The results of Tables~\ref{tab:random-tests} and~\ref{tab:random-tests-one-patch-per-column} suggest that the expected code size is much smaller than that achievable by our algorithm.
There appears to be large variance in the code sizes achieved by generating codes at random.
However, even the maximum among 100 trials, in nearly all cases, fell far below the code sizes given by the algorithm of Theorem~\ref{thm:polynomial-construction}.

We emphasize that Theorem~\ref{thm:random-construction} applies only to the most na\"{i}ve way to generate random codes.
It does not rule out that better performance may be obtained by a more sophisticated strategy.
For example, a greedy algorithm that generates macrobonds at random
until a new one appears that has low auto-correlation,
and low cross-correlation with existing macrobonds
(rather than quitting upon encountering the first ``bad'' macrobond),
would outperform the strategy above.
A more sophisticated stochastic local search may perform even better.
Our goal in this section is not to find the ``best'' randomized method, but merely to demonstrate that the deterministic algorithm of Theorem~\ref{thm:polynomial-construction} has better performance than the simplest imaginable randomized algorithm.
It is also the case that we were able to theoretically analyze the random codes of Theorem~\ref{thm:random-construction}, but do not know how to do this for more sophisticated randomized algorithms.
Often random codes perform quite well,
depending on the task,
so it is curious that in this case, they do not.

\section{Open Questions}\label{sec:open-questions}

A number of directions for future work suggest themselves.

\begin{enumerate}
  \item 
  We chose to define a macrobond as a subset of an $n \times n$ square for convenience, and because it worked well with our proof technique using polynomials over finite fields.
  An obvious generalization is to find geometric orthogonal codes that work over $n \times m$ rectangles for $n \neq m$.
  Of course, one can simply add/remove empty rows/columns without altering the auto- or cross-correlation, but is there a technique for generating macrobonds that is ``naturally'' defined over a rectangle or other geometry?

  \item \label{open-question:other-w}
  Our lower bound technique works for $w=n$, where $w$ is the desired number of patches per macrobond.
  Can we generalize to arbitrary $w$?
  The obvious way to generalize to larger values of $w$ is to assign several polynomials $p_1,\ldots,p_k$ (where $k < n$) to each macrobond $M$, defining $M = \setr{(x,p_i(x))}{x \in \GF{n}, 1 \leq i \leq k}.$
  Overlap between different polynomials in the same macrobond (i.e., $p_i(x)=p_j(x)$ for some $x\in\GF{n}$ and $1\leq i<j\leq k$) could result in up to $(k-1)\lambda$ total points of overlap, so rather than having $n k$ patches in the macrobond, some could end up with as few as $n k - (k - 1)\lambda$ patches.
  By removing arbitrarily chosen points from macrobonds with more points than this, we could assume all have exactly $w = k n - (k - 1)\lambda$ patches.
  A straightforward modification of the proof of  Theorem~\ref{thm:polynomial-construction} then shows that choosing degree-$\lambda$ polynomials results in an $(n, n k - (k - 1)\lambda, \lambda k^2)$-geometric orthogonal code of size $\lfloor (n^{\lambda-1} - n^{\lambda-2})/k \rfloor$.\footnote{Essentially the same proof works, but now each of the $k^2$ pairs of polynomials between two macrobonds could contribute $\lambda$ overlapping patches, resulting in up to $\lambda k^2$ total overlapping patches. However, each macrobond now has $k>1$ polynomials associated to it, so with $n^{\lambda-1} - n^{\lambda-2}$ total polynomials we get $\lfloor (n^{\lambda-1} - n^{\lambda-2})/k \rfloor$ total macrobonds.}
  Is there another way to generalize to other values of $w$,
  including $w$ that are not of the form $kn - (k-1)\lambda$?


  \item
  Reduce the upper bound of Theorem~\ref{thm:upper-bound}.

  \item
  Increase the lower bound of Theorems~\ref{thm:polynomial-construction} or~\ref{thm:polynomial-construction-flipping}.

  \item 
  Generalize from primes to arbitrary $n\in\Z^+$.

  \item
  The upper bound on random codes of Theorem~\ref{thm:random-construction} applies to macrobonds selected uniformly at random from the set of macrobonds that have one patch per column. Generalize to macrobonds selected uniformly at random from the set of \emph{all} macrobonds of a fixed weight $w$.

  \item
  Decrease the upper bound on random codes proven in Theorem~\ref{thm:random-construction}.
  The large difference between the medians and the $U_r(n,\lambda,\frac{1}{2})$ numbers in rows of Tables~\ref{tab:random-tests} and~\ref{tab:random-tests-one-patch-per-column} suggest that the $U_r(n,\lambda,\frac{1}{2})$ bound is not tight.

  \item 
  In defining orthogonality of two macrobonds, we allow them to translate relative to each other and to rotate, but only by 180$^\circ$.
  A 2D macrobond based on generalizing the scheme of Figure~\ref{fig:origami}(d) in the most obvious way, in which the blunt ends face orthogonal to the origami face rather than parallel to it as in Figure~\ref{fig:origami}(e), would not have a patch shape that automatically disallows non-180$^\circ$ rotations.\footnote{As mentioned in Section~\ref{sec:intro:statement-of-result}, there are physical reasons to conjecture that such rotations have weaker stacking bonds than the ``standard'' rotation. } 
  Thus, it would be interesting to consider adding a rotational constraint to the definition of geometric orthogonal code.
  For such a macrobond, it would make sense to consider overlaps when a rotation brings points ``close'' to each other, even if not exactly overlapping.
  For example, perhaps patch pairs separated by distance at least~1 are far enough apart that they cannot bind (even with distortion as seen in Fig.~\ref{fig:origami}(f)), but those of distance less than~1 (even if distance $>0$) could bind.
  Then two macrobonds have correlation $> \lambda$ if some relative translation and rotation of them brings $> \lambda$ patch pairs to strictly less than distance~1 from each other.
  In other words, patches behave as diameter-$1$ circles moving continuously, rather than as width-1 squares moving discretely by integer distances, and correlation corresponds to the number of overlapping circular patches between two macrobonds.

  \item 
  We model patches as completely non-specific bonds.
  DNA blunt ends are \emph{relatively} nonspecific, but even so, a GC/CG stack, for instance, is significantly stronger than an AT/TA stack.
  The macrobonds employed in~\cite{sungwook_bonds_2011} use only GC/CG stacks to enforce uniformity, but other stack types are allowed in~\cite{dietz_bonds_science_2015}.
  One can imagine ways to add some specificity to patches by choice of terminating base pair, or possibly by using DNA sticky ends in place of stacking bonds.
  The problem is then more accurately modeled by defining a macrobond to be a function $M:[n]^2 \to C \cup \{\mathsf{null}\}$, where $C$ is a finite set of ``colors'', and $\mathsf{null}$ represents the absence of a patch.
  Then, two aligned patches with colors $c_1,c_2 \in C$ have strength $\mathsf{str}(c_1,c_2)$ (for $C$ being the set of possible terminating DNA base pairs, $\mathsf{str}(c_1,c_2)$ is in Table 1 of\cite{santalucia2004thermodynamics}).

  \item
  \label{open-question:multiple-w}
  Complementary to the previous question, we observed in Section~\ref{sec:intro:relationship-defn-experiments} that some experimental work requires bonds to have unequal strength.
  If all patches in a macrobond have the same strength, this would correspond to using different numbers of patches in different macrobonds.
  There are many interesting variations of this question, but here is a concrete open question:
  Suppose that we want exactly two strengths of macrobonds, $w_1$ and $w_2 \neq w_1$.
  Give an algorithm that, on input $n,w_1,w_2,\lambda,s_1,s_2 \in \Z^+$, outputs $s_1$ macrobonds (subsets of $n \times n$ square) of weight $w$ and $s_2$ macrobonds of weight $w_2$, such that all macrobonds have auto-correlation and pairwise cross-correlation at most $\lambda$, or the algorithm reports that no such macrobonds exist.
  (Note that in this setting there is no ``largest'' code:
  one may be able to trade off $s_1$ and $s_2$ in several Pareto-optimal ways.)

  \item
  \label{open-question:reflect-one-axis}
  Our formalization of the concept of macrobonds resembles Figure~\ref{fig:origami}(e) more than~\ref{fig:origami}(d) in the sense that there are two types of faces (``bump'' type faces and ``dent'' type faces), and a macrobond is always formed between opposite-type faces.
  In contrast, macrobonds formed in Figure~\ref{fig:origami}(d) are between faces of the same ``type''.
  In this case, one could imagine a macrobond coming into contact with a reflected copies of other macrobonds, rather than rotated copies as captured by our geometric 180-rotating orthogonal codes.
  Is a result similar to Theorem~\ref{thm:polynomial-construction-flipping} possible for this definition of ``flipping'' a macrobond?

  \item
  If we think of the 1D edge of an origami as vertical, then all patches lie at $x=0$ (since the edge is vertical), and a macrobond chooses a subset of $y$ values at which to place patches.
  Woo and Rothemund~\cite{sungwook_bonds_2011} study a related technique for creating specific macrobonds,
  in which patches are placed at \emph{all possible} $y$ values along the edge, but modifies the \emph{shape} of the edge itself so that some patches lie at different $x$ values; see Fig.~3 of~\cite{sungwook_bonds_2011}.
  This sterically prevents all patches from bonding unless the shapes are complementary and aligned properly.
  So for $n$ patch locations, a code is specified by function $c: [n] \to \{0,1,\ldots,d\}$, where $d \in \N$ is the maximum allowed ``depth'' (position along the $x$-axis) of a patch, relative to the patches that are furthest away from the center of the origami, which defined to be at depth 0.
  It would be interesting to prove upper and lower bounds on code sizes based on $n$, $d$, and threshold $\lambda$ (see also~\cite{InformationCapacityofSpecificInteractionsPNAS}).

  \item
  Instead of modeling the macrobond as a discrete set of points, model it as a subset $M \subset \R^2$ of the plane.
  A similar setting was considered by Gopinath, Kirkpatrick, Rothemund, and Thachuk~\cite{gopinath2016progressive}, who studied the following problem motivated by DNA origami experiments:
  Design a \emph{shape} $S$ (a bounded, connected subset of $\R^2$)
  and a ``target'' shape $T$ (possibly $T=S$ but not necessarily),
  such that, starting from any initial placement (translation and rotation) of $S$ having with non-zero overlap with $T$,
  there is a continuous rigid motion taking $S$ to a \emph{unique} placement that (globally) maximizes the area of overlap between $S$ and $T$, such that the motion has a monotonically increasing overlap
  (i.e., there is no local maximum or plateau of suboptimal overlap in which $S$ can get ``stuck'').
  An interesting open problem is to find several (or even \emph{two}) shape-target pairs that have this property, but that also have low cross-correlation.
  (In this setting auto-correlation is not a concern, since the lack of local maxima or plateaus implies that misaligned translations between a shape and its own target will correct themselves by re-alignment.)

\end{enumerate}

\noindent{\bf Acknowledgements.}
{ 
    We are grateful to Matt Patitz for organizing the 2015 University of Arkansas Self-assembly Workshop, where this project began, and to the participants of that workshop and support from NSF grant CCF-1422152.
    We thank the participants of the 2015 Workshop on Coding Techniques for Synthetic Biology at the Univ. of Illinois, Urbana-Champaign, especially Han Mao Kiah, Farzad Farnoud, Urbashi Mitra, and Olgica Milenkovic, for bringing optical orthogonal codes to our attention and giving valuable feedback.
    We thank Paul Rothemund and Sungwook Woo for explaining known physical properties of stacking bonds.
    %
    %
    We thank anonymous reviewers for pointing out mistakes in an early draft, and in particular for several very helpful suggestions for cleaning up the presentation.
    We are indebted to Ray Li for correcting several mistakes in the conference version of this paper, in particular for showing how our proof technique for Theorems~\ref{thm:polynomial-construction} and~\ref{thm:polynomial-construction-flipping} works only for prime $n$, not prime powers, and for observing that we require $a_0=0$ in each proof as well, changing the bounds in those theorems, and for pointing out that using multiple polynomials per macrobond allows one to generalize to larger values of $w$, at the cost of potentially having to remove a small number of overlapping points, as discussed in Open Question~\eqref{open-question:other-w}.
}

\bibliographystyle{plain}
\bibliography{geometric-bonds}

\end{document}